\documentclass[10pt, compsoc, onecolumn]{IEEEtran} 

\usepackage{bm}
\usepackage{amsmath, amssymb,MnSymbol }
\usepackage{array}
\usepackage[T1]{fontenc}
\usepackage{amssymb,amsfonts, amsthm}
\usepackage{mathrsfs, relsize}
\usepackage{mathtools}
\usepackage{algorithm}
\usepackage{algpseudocode}
\usepackage[numbers]{natbib}
\usepackage[dvipsnames]{xcolor}
\usepackage{algorithm}
\usepackage{bm}
\usepackage{amsmath, amssymb,MnSymbol }
\usepackage{array}
\usepackage[T1]{fontenc}
\usepackage{amssymb,amsfonts}
\usepackage{mathrsfs, relsize}
\usepackage{mathtools}
\usepackage{algorithm}
\usepackage{enumerate}
\usepackage{algpseudocode}
\usepackage[numbers]{natbib}
\usepackage{color}
\usepackage[dvipsnames]{xcolor}
\usepackage{algorithm}

\newcommand{\RR}{\mathbf{R}}

\renewcommand{\citet}{\cite}

\renewcommand{\phi}{\varphi}
\renewcommand{\a}{\bm{a}}
\renewcommand{\b}{\bm{b}}
\renewcommand{\c}{\bm{c}}
\renewcommand{\d}{\bm{d}}

\renewcommand{\u}{\bm{u}}
\renewcommand{\v}{\bm{v}}
\newcommand{\w}{\bm{w}}
\newcommand{\x}{\bm{x}}
\newcommand{\y}{\bm{y}}
\newcommand{\z}{\bm{z}}

\newcommand{\0}{\bm{0}}
\newcommand{\1}{\bm{1}}

\newcommand{\A}{\mathbf{A}}
\newcommand{\B}{\mathbf{B}}

\newcommand{\R}{\mathbf{R}}

\newcommand{\T}{\mathbf{T}}

\newcommand{\setJ}{\mathcal{J}}

\newcommand{\setP}{\mathcal{P}}

\newcommand{\setR}{\mathcal{R}}

\newcommand{\argmin}{\operatornamewithlimits{argmin}}

\newtheorem{claim}{Proposition}
\theoremstyle{definition}
\newtheorem{definition}{Definition}

\renewcommand{\T}{{\mbox{\sc t}}}

\usepackage[utf8]{inputenc}
\usepackage{wrapfig}
\usepackage{graphicx}
\usepackage[normalem]{ulem}
\usepackage{caption, multicol, balance}

\newtheorem{remark}{Remark}
\usepackage{url}
\newcommand{\Add}[1]{{\color{black}{#1}}}

\newcommand{\Delete}[1]{}
\title{Multi-Path Alpha-Fair Resource Allocation at Scale in Distributed Software Defined Networks$^{\star}$\footnote{$^{\star}$This is the authors' edited copy of the article ``{Multi-Path Alpha-Fair Resource Allocation at Scale in Distributed Software Defined Networks}" to appear in \emph{IEEE JSAC Special Issue on Scalability Issues and Solutions for Software Defined Networks} (accepted August 2018).}}

\author{  \IEEEauthorblockN{Zaid Allybokus\IEEEauthorrefmark{1}\IEEEauthorrefmark{2}, Konstantin Avrachenkov\IEEEauthorrefmark{1}, Jérémie Leguay\IEEEauthorrefmark{2}, Lorenzo Maggi\IEEEauthorrefmark{2}}
    \\\IEEEauthorblockA{\IEEEauthorrefmark{2}Huawei Technologies, France Research Center
    \\\{zaid.allybokus, jeremie.leguay, lorenzo.maggi\}@huawei.com}
    \\\IEEEauthorblockA{\IEEEauthorrefmark{1}INRIA Sophia Antipolis
    \\ konstantin.avratchenkov@inria.fr}
 }        

\begin{document}
\maketitle

\begin{abstract}
The performance of computer networks relies on how bandwidth is shared among different flows. Fair resource allocation is a challenging problem particularly when the flows evolve over time. 
To address this issue, bandwidth sharing techniques that quickly react to the traffic fluctuations are of interest, especially in large scale settings with hundreds of nodes and thousands of flows. In this context, we propose a distributed algorithm based on the Alternating Direction Method of Multipliers (ADMM) that tackles the multi-path fair resource allocation problem in a distributed SDN control architecture. Our ADMM-based algorithm continuously generates a sequence of resource allocation solutions converging to the fair allocation while always remaining feasible, a property that standard primal-dual decomposition methods often lack. Thanks to the distribution of all computer intensive operations, we demonstrate that we can handle large instances at scale.
\end{abstract}

\section{Introduction}

Software Defined Networking (SDN) technologies are radically transforming network architectures by offloading the control plane (e.g., routing, resource allocation) to powerful remote platforms that gather and keep a local or global view of the network status in real-time and push consistent configuration updates to the network equipment. The computation power of SDN controllers fosters the development of a new generation of control plane architecture that uses compute-intensive operations. 
Initial design of SDN architectures \cite{vaughan2011openflow} had envisioned the use of one central controller. However, for obvious scalability and resiliency reasons, the industry has quickly realized that the SDN control plane needs to be partially distributed in large network scenarios \cite{kreutz2015software}. 
Hence, although \emph{logically} centralized, in practice the control plane may consist of multiple controllers each in charge of a SDN domain of the network and operating together, in a \emph{flat} \cite{stallings2013openflow} or \emph{hierarchical} \cite{hassas2012kandoo} architecture. In Fig.~\ref{fig:architecture}, an example of flat architecture is illustrated. In hierarchical architectures, a master controller is placed on top of domain sub-controllers and keeps global view of the network and the sub-controllers to handle the message-passing protocols. 

\begin{wrapfigure}{i}{0.5\textwidth}
  \begin{center}
    \includegraphics[width=0.48\textwidth]{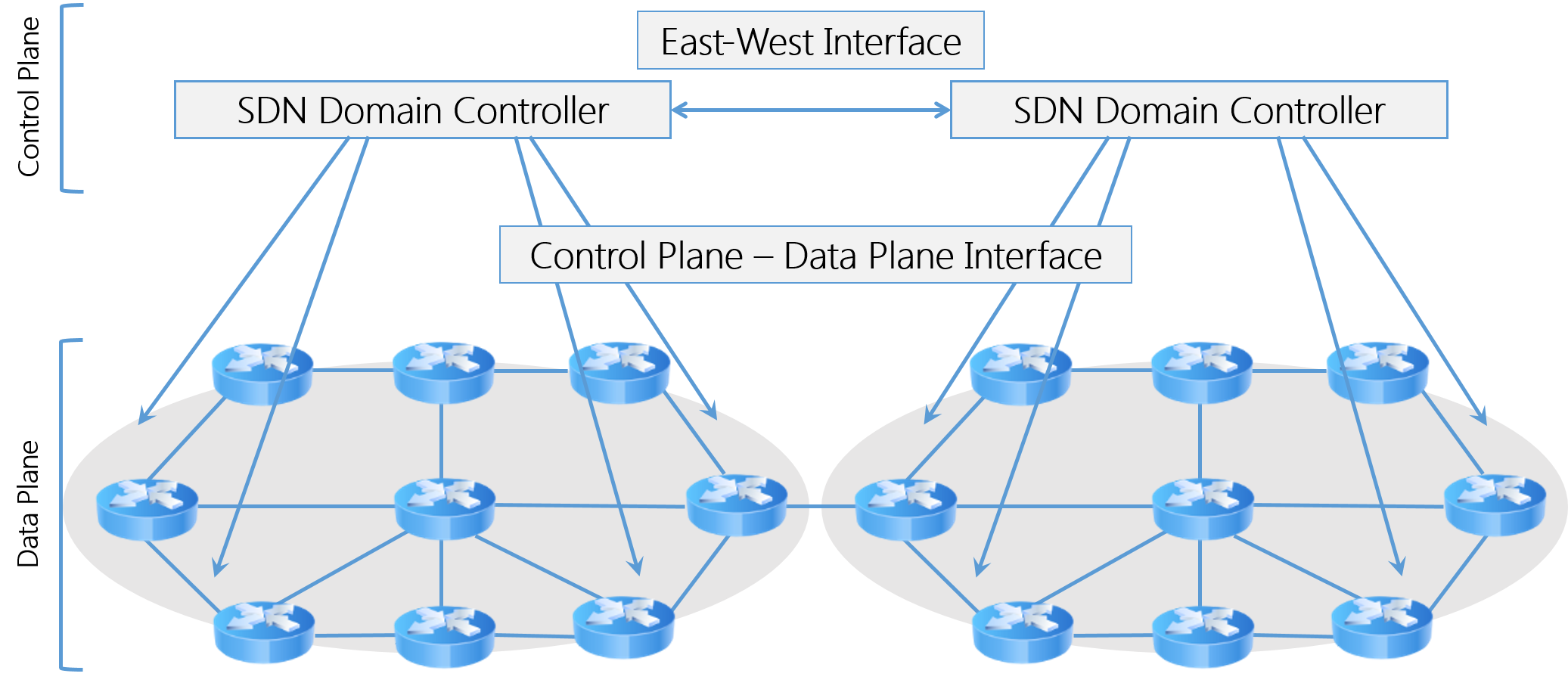}
  \end{center}
  \caption{Distributed SDN architecture.}
  \label{fig:architecture}
\end{wrapfigure}
In this paper, we study the problem of computing a globally \emph{fair}  (in the sense of $\alpha$-fairness defined by Mo and Walrand in \cite{mo2000fair}, see Section~\ref{sec:problem}) multi-path resource allocation in a distributed SDN scenario, where the control plane is distributed over several domain controllers. In this context, flows transiting in the network typically correspond to traffic aggregates of a customer or a class of customers. \Add{The traffic aggregates can be carried on one path or be split through many different paths connecting a source to a destination. Indeed, the multi-path traffic engineering can be preferred over single-paths to ensure routing robustness, low latency, or good load balancing for better performance, as explained in \cite{kumar2018semi}}.

In this paper, we study the problem of computing a globally \emph{fair} (in the sense of $\alpha$-fairness defined by Mo and Walrand in \cite{mo2000fair}, see Section~\ref{sec:problem}) multi-path resource allocation in a distributed SDN scenario, where the control plane is distributed over several domain controllers. In this context, flows transiting in the network typically correspond to traffic aggregates of a customer or a class of customers.
We consider the traffic engineering use case where the size of flows evolves over time and the bandwidth reserved to each of them has to be quickly adjusted towards the novel fair solution.   


In distributed SDN architectures \Add{\cite{chang2016asynchronous}\cite{phemius2014disco}}, each controller has full information about its own domain. \Add{Although for fault-tolerance reasons, they could be composed of master and slave agents that act as a single entity \cite{obadia2014failover}, in this paper we only consider the distribution of the SDN control plane for scalability purposes: the control plane consists of multiple controllers each in charge of a single domain and a reduced portion of the global workload and operate together as a logically centralized controller. Therefore}, each controller can communicate with adjacent peer controllers and/or with a central, upper-layer controller entity. 
However, exchanges between controllers are expensive in terms of communication delay and overhead~\cite{phemius2014disco}. This technological limitation translates directly into an algorithmic constraint: distributed algorithms for SDN have a limited budget in terms of the number of iterations to reach convergence, \Add{i.e. a near-optimal solution}.

In distributed SDN architectures, as depicted in Fig.~\ref{fig:architecture}, each controller has full information about its own domain. Moreover, it can communicate with adjacent peer controllers and/or with a central, upper-layer controller entity. 
However, exchanges between controllers are expensive in terms of communication delay and overhead~\cite{phemius2014disco}. This technological limitation translates directly into an algorithmic constraint: distributed algorithms for SDN have a limited budget in terms number of iterations to reach convergence.

A second crucial property for any distributed algorithm for SDN is responsiveness. In fact, the network state may be affected by abrupt changes, e.g., flow size variation, flow arrival/departure, link/node congestion. In this case, convergence for the previous network state may not even be attained when a change occurs in the system. For this reason, it is often preferable to have a quick access to a good quality solution rather than a provably asymptotically optimal solution with poor convergence rate. Hence, it is crucial that the resource allocation computed by a distributed algorithm is \emph{feasible, thus implementable}, at any iteration.

Also, modern SDN controllers \cite{bilal2016overview} rely on grid computing technologies \Add{such as Akka \cite{akka} or Hazelcast~\cite{shin2017iris}, respectively for the two major open source SDN controllers OpenDayLight \cite{opendaylight} and ONOS \cite{onosproject}. As a consequence, 
massively parallelizable algorithms for SDN should be preferable as more adapted and better likely to tackle scalability issues}. 

Also, modern SDN controllers rely on grid computing technologies such as Akka or Hazelcast~\cite{shin2017iris} for compute intensive tasks. As a consequence, any distributed algorithm for SDN should be massively parallelized. 

To recap, we identify two main requirements for a distributed algorithm for fair resource allocation, namely  \textit{i)}~converging to a ``good'' fair solution in a small number of iterations and \textit{ii)} producing \emph{feasible} solutions at all iterations.

To recap, we identify three main requirements for a distributed algorithm for fair resource allocation, namely  \textit{i)}~converging to a ``good'' fair solution in a small number of iterations, \textit{ii)} producing \emph{feasible} solutions at all iterations and \textit{iii)} being massively parallelized.

We claim that none of the current methods that allocate resources in an SDN scenario is able to achieve the three aforementioned goals at the same time. 
Local mechanisms such as Auto-Bandwidth \cite{dhody2009autobandwidth} have been proposed to greedily and distributedly adjust the allocated bandwidth to support time-varying IP traffic in \emph{Multi Protocol Label Switching} (MPLS) networks. Auto-Bandwidth successfully tackles goals \textit{ii)} and \textit{iii)}, but not \textit{i)}, as it neither ensures fairness nor optimizes resources globally. 
Also, classic primal-dual algorithms have been proposed to solve the $\alpha$-fair resource allocation problem in distributed SDN scenarios, as in \cite{mccormick2014real}. However, primal-dual algorithms are known to fail at providing feasible solutions at any iteration step, thus they fail at achieving goal \textit{ii)}.

Recently in the optimization research community, the \emph{Alternating Direction Method of Multipliers} (ADMM) \cite{boyd2011distributed} has captured the attention for its separability and fast convergence properties. 
We claim that ADMM offers new and yet unexploited possibilities to tackle concurrently the goals \textit{i)}, \textit{ii)} and \textit{iii)}. Indeed, in this paper we show how ADMM serves our purposes, by allowing all controllers to handle their own domains simultaneously, while still converging to a global optimum in the fashion of a general distributed consensus problem.

\noindent \textbf{Main contributions}: We develop Fast Distributed ADMM (FD-ADMM, Algorithm \ref{FDADMM}) for the multi-path $\alpha$-fair resource allocation problem over a distributed SDN control plane. It iteratively produces resource allocations that converge to the $\alpha$-fair optimal allocation. Heavy computations, requiring projections on polytopes, can be massively parallelized on a link-by-link basis (Algorithm \ref{FDADMM}, line 7) in each domain. This yields a convergence rate (in terms of iteration count) that does not depend on the partitioning of the network into domains, that can therefore be done \emph{independently}.

We showed \cite{allybokus2017real} that our FD-ADMM algorithm can function in real-time, as \textit{i)} close-to-optimal solutions are available since the \emph{very first} iterations and \textit{ii)} feasible allocations are available at \emph{all} iterations (Proposition \ref{claim:feasible}), a property that standard primal-dual decomposition methods generally lack. This permits to adjust within very short time the bandwidth of flows that evolve quickly and need immediate response.\\
Moreover, we addressed in \cite{allybokus2017lower} the problem of the penalty parameter initialization in ADMM that is well-known to highly condition the convergence speed of the algorithm, by proposing a tuning based on a lower bound for the $\alpha$-fair resource allocation problem that we derived.

This article is a synthesis of the two mentioned works \cite{allybokus2017lower} and \cite{allybokus2017real} with three novel contributions:

\begin{itemize}
\item We extend the model presented in \cite{allybokus2017real} to the setting of \emph{multi-path} routing. Now, each flow can carry their traffic along several paths instead of only one. 
\item We also show that the algorithm can integrate a switching cost in order to address the related relevant problem of limiting flow reconfigurations. Although we expect FD-ADMM to continuously provide feasible iterates that respond to traffic variations in real-time, it is practically infeasible to reconfigure all the flows too often without harming the network stability and the overhead \cite{paris2016controlling}. Therefore, we use ideas of sparse optimization \cite{bach2012optimization} to make FD-ADMM sensitive to the cost of a reconfiguration of the bandwidth along a path, and hence operate a trade-off between fairness and switching cost.
\item Finally, we evaluate numerically the performance of FD-ADMM over large scale instances made of \Add{Barabasi-Albert and Fat tree} networks of up to hundreds of nodes and thousands of links, requests and paths, in a setting where several SDN domain controllers operate in parallel.  
\item \Add{To the best of our knowledge, we were the first to show how ADMM can help designing real-time distributed algorithms for computing $\alpha$-fair resource allocations in distributed settings. We were also the first to address the penalty parameter tuning of ADMM in this situation. Those results are all extended to the multi-path setting in the present article. }
\end{itemize}



The remainder of this paper is organized as follows. Section~\ref{sec:related} surveys the related work around the fair resource allocation problem. Section~\ref{sec:problem} formulates the multi-path $\alpha$-fair resource allocation problem and recalls the key ideas of ADMM. Section~\ref{sec:consensus} introduces FD-ADMM, our distributed ADMM-based algorithm that benefits from the distribution of SDN controllers over multiple domains, that relies on a reformulation of the problem in the fashion of a general distributed consensus problem. Section~\ref{sec:reconfiguration} extends the model to account for the introduction of the switching cost. Section~\ref{sec:execution} provides large scale simulations that validate our approach and finally, Section~\ref{sec:conclusion} concludes the paper.

\section{Related work}
\label{sec:related}

The concept of fair resource allocation has been a central topic in networking. Particularly, \emph{max-min} fairness has been the classic resource sharing principle \cite{bertsekas1992data} and has been studied extensively. The concept of \emph{proportional fairness} and its weighted variants were introduced in \cite{kelly1998rate}. Later, a spectrum of fairness metrics including the two former ones was introduced in \cite{mo2000fair} as the family of $\alpha$-fair utility functions.

Some early notable works on max-min fairness include \cite{charny1995congestion}, where the authors propose an asynchronous distributed algorithm that communicates explicitly with the sources and pays some overhead in exchange for more robustness and faster convergence. Later in \cite{skivee2004distributed}, a distributed algorithm is defined for the weighted variant of max-min fair resource allocation problem in MPLS networks, based on the well-known property that an allocation is max-min fair if and only if each \emph{Label-Switched Path} (LSP) either admits a \emph{bottleneck link} amongst its used links or meets its maximal bandwidth requirement (see Definition 4 there of a bottleneck link). The problem of Network Utility Maximization (NUM) was also addressed with standard decomposition methods that could give efficient and very simple algorithms based on gradient ascent schemes performing their update rules in parallel. In this context, Voice \cite{voice2006stability}, then McCormick et al. \cite{mccormick2014real}, tackle the $\alpha$-fair resource allocation problem with a gradient descent applied to the dual of the problem. 

In the works described above, no mention is made on the potential (in fact, systematic) feasibility violation of the sequences generated by those algorithms, which is a crucial matter in  distributed SDN settings. Regarding this topic, the authors of \cite{lee1998decentralized} employ damping techniques to avoid transient infeasibility while reaching the max-min fair point, but cannot guarantee feasibility at all times, especially in dynamic settings. 
Also motivated by this, more recently the authors of \cite{sundaresan2016iterative} provide a feasibility preserving version of Kelly's methodology in \cite{kelly1998rate}. Their algorithm introduces a slave that gives at each (master) iteration an optimal solution of a weighted proportionally fair resource allocation problem that is explicitly addressed in only the two cases of polymatroidal and flow aggregating networks. In fact, our paper contributes to this problem by proposing an efficient \emph{real-time} version of the slave process, for any topology, preserving feasibility at each (slave) iteration. Amongst approximative approaches, one can quote the very recent work \cite{marasevic2016fast} where a multiplicative approximation for $\alpha \neq 1$ and additive approximation for $\alpha = 1$ is provably obtained in poly-logarithmic time in the problem parameters. Moreover, starting from any point, the algorithm reaches feasibility within poly-logarithmic time and remains feasible forever after. The algorithm described in our paper solves the problem optimally and reaches feasibility as from the first iteration from any starting point.

The work around ADMM is currently flourishing. The $\mathcal{O}(\frac{1}{n})$ best known convergence rate of ADMM \cite{he2012douglas} failed to explain its empirical fast convergence until very recently in \cite{Deng2016}, where global linear convergence rates are established in four scenarios of the strongly convex case. ADMM is also well-known for its performance that highly depends on the parameter tuning, namely, the penalty parameter $\lambda$ in the augmented Lagrangian formulation (see Section \ref{augmentedlag} below). An effective use of this class of algorithms cannot be decoupled from an accurate parameter tuning, as convergence can be extremely slow otherwise. Thus, in the same paper  \cite{Deng2016}, the authors provide a linear convergence proof that yields a convergence rate in a closed form that can be optimized with respect to the problem parameters. Therefore, thanks to these works, we derived in \cite{allybokus2017real} an adaptive tuning of ADMM, and in \cite{allybokus2017lower} a satisfactory initialization of the penalty parameter for the $\alpha$-fair resource allocation problem. Several papers use the distributivity of ADMM to design efficient distributed algorithms solving consensus formulations for e.g. model predictive control \cite{mota2012distributed} and resource allocation in wireless virtual networks \cite{liang2015distributed} but do not address this fundamental detail.





\section{Multi-path fair resource allocation problem}
\label{sec:problem}
In this section, we define the multi-path $\alpha$-fair resource allocation problem and formulate it as a centralized convex optimization problem. We also introduce the basic principles of ADMM and then present the centralized algorithm it yields. We will see in what this centralized version does not fit our distributed setting, which motivates the more detailed decomposition of this article.  

\subsection{Presentation}
\label{subsec:problemformulation}
We define the network as the set of its links $\setJ$. Each network link $j\in \setJ$ has a total capacity of $c_j \in \R_+$.
 
Let $\setR$ be a set of \emph{connection requests}, or shortly, \emph{requests}, over the network. We regard each request $r$ as a communication demand between a source node and a destination node of the network, provided with a set $P_r$ of several paths connecting the two end nodes. We assume the paths are pre-computed once and for all and do not change. In this article, only the bandwidth allocation along fixed paths is considered and as such, the path computation dimension of the problem goes beyond the scope of this study. We model the paths as subsets of $\setJ$, that we assume form a physical path on the network's real topology.  
Thus, in the multi-path setting, the link capacities will be shared by the different requests in $\setR$ with a bandwidth allocation along one or many of their paths. Let $\setP \coloneqq \cup_r P_r$ be the set of all established paths. For each path $p\in\setP$, we define $J_p$ as the set of links that form $p$. For example, if the path $p$ contains the set of links $j_1, j_2$ and $j_3$, then $J_p = \{j_1,j_2,j_3\}$.

The aim is to compute an optimal bandwidth allocation with respect to the  $\alpha$-fairness metric
defined in \cite{mo2000fair}, that we report below.

\begin{definition}[$(w,\alpha)$-fairness, \cite{mo2000fair}]
\label{defalphafair}
Let $\alpha\in \R_+$. Let $F\subset \R_+^n$ be a non-empty feasible set not reduced to $\{\0\}$.
Let $\w \in \R_+^n$ and $\y^*\in F$. We say that $\y^*$ is $(\w,\alpha)$-\emph{fair} (or simply $\alpha$-fair when there is no confusion on $\w$) if the following holds:
\[\forall r\in [1,n],\quad y^*_r >0\quad\mbox{and}\quad\forall \y\in F 
,\quad \sum_{r=1}^n w_r\frac{y_r - y^*_r}{y^{*\alpha}_r} \leq 0. \]
Equivalently, $\y^*$ is $(\w,\alpha)$-fair if, and only if $\y^*$ maximizes the $\alpha$-fair utility function $f^\alpha$ defined over $F-\{0\}$:
\[f^\alpha(\y) = \sum_{r = 1}^n f^\alpha_r(y_r),
\]
where
\[
f_r^\alpha(y_r) =
\left\{ \begin{array}{ll}
w_r \log(y_r) & \mbox{ if } \alpha=1,\\
w_r \frac{y_r^{1-\alpha}}{1-\alpha} & \mbox{ otherwise. }\\
\end{array}\right. 
\]
\end{definition}
The success of $\alpha$-fairness is due to its generality: in fact, for $\alpha=0,1,2,\infty$ it is equivalent to max-throughput, proportional fairness, min-delay, and max-min fairness, respectively \cite{srikant2013communication}. The introduction of the weight vector $\w$ permits to assign to a request a level of priority: the larger $w_r$, the higher the system's incentive to allocate the available bandwidth to request $r$. In general, $w_r$ can model the number of connections of the same type of a connection request $r$ or the amount of backlogged traffic for this request, for example.

\noindent {\bf Sharing policy.} We define the \emph{link-path} incidence matrix $\A \in \R^{|\setJ|\times |\setP|}$, and the \emph{request-path} incidence matrix $\B \in \R^{|\setR|\times |\setP|}$, as the following:
\begin{equation}
A_{jp} = \left\{\begin{array}{ll} 1 &\mbox{if } j\in J_p\\ 0& \mbox{otherwise}\end{array}\right., \mbox{ and } B_{rp} = \left\{\begin{array}{ll} 1 &\mbox{if } p\in P_r\\ 0& \mbox{otherwise.}\end{array}\right.
\end{equation}

The multi-path $\alpha$-fair bandwidth sharing policy can be defined as follows. For each request $r$, the network can attribute a bandwidth of $x_p \geq 0$ along one or more of its established paths $p\in P_r$. The \emph{path-wise allocation} $\x$ is then defined as the vector $(x_p)_{p\in \setP}$. The \emph{aggregate bandwidth} of request $r$, $y_r \geq 0$, represents the total bandwidth allocated to request $r$, that is, 
\begin{equation}y_r \coloneqq (\B\x)_r = \sum_{p\in P_r} x_p. \label{aggregate_eq}\end{equation}
The goal of the network is to find a path-wise allocation $\x$ in such a way that the aggregate bandwidth vector $\y \coloneqq (y_r)_{r\in \setR}$ maximizes the $\alpha$-fair metric.  The path-wise allocation $\x$ should also respect the link capacity constraints of the network, that read:
\begin{equation}
\label{eq:capconU}
\A\x \leq \c \Leftrightarrow  \sum_{p : j \in J_p } x_p \leq c_j \quad \forall j\in \setJ.
\end{equation}
Therefore, we have the problem definition below:

\begin{definition}\label{defmultipath}The \emph{multi-path $\alpha$-fair allocation problem} is the problem of finding a path-wise bandwidth allocation $\x$ (and therefore its corresponding aggregate bandwidth allocation $\y = \B\x$) such that the function $f^{\alpha}$ is maximized with respect to $\y$ and the link capacity constraints $\A\x \leq \c$ are respected.
\end{definition}
It is well-known that the function $f^\alpha$ admits a unique maximizer over any convex closed bounded set. In our case, this means that the optimal aggregate bandwidth allocation $\y^*$ is unique. We draw the reader's attention to the fact that the fairness is measured upon the aggregate bandwidth allocation $\y$, not upon the path-wise allocations: one request may be allocated resources along several paths which do not need to be ``fair'' to each other. Remarkably, for the unique optimal aggregate bandwidth allocation $\y^*$, there may\footnote{The unicity of $\y^*$ follows from the strict convexity of $f^\alpha$ as a function of $\y$. As the function $\x \mapsto f^\alpha(\B\x)$ is \emph{not generically} strictly convex with respect to the path-wise variable $\x$, the unicity of a path-wise optimum is no more guaranteed.} be several path-wise allocations that verify the equation $\B\x = \y^*$.
\subsection{Problem formulation}
In the rest of this article, we adopt the convex optimization terminology. Define for each $r\in \setR$ the convex cost function $g_r^\alpha(y_r)\coloneqq -f_r^\alpha(y_r)$. Then,  $g^\alpha(\y) \coloneqq \sum_{r\in \setR} g_r^\alpha(y_r)$ is a convex, closed and proper\footnote{The term \emph{closed} stands for lower semi-continuous and \emph{proper} means non-identically equal to $\infty$.} function over $\RR_+^{|\setR|}$. 
%
We introduce $\iota$ as the \emph{convex indicator function} of the capacity constraints:
$$\iota(\x) =
\left\{ \begin{array}{ll}
0 & \mbox{if }\A\x\leq \c, \x\geq 0\\
\infty, & \mbox{otherwise}.
\end{array}\right.$$ 
Then the multi-path $\alpha$-fair allocation problem (Def.~\ref{defmultipath}) can equivalently be formulated as the following convex program:

\begin{equation}
\min_{\y,\x} \sum_{r\in \setR}g^\alpha_r\left(y_r\right) + \iota(\x),
\label{ConvObj}
\end{equation}
\begin{equation}
\mbox{s.t.}\quad \y-\B\x = 0.
\label{ConvConst}
\end{equation}

\subsection{ADMM as an augmented Lagrangian splitting}
\label{augmentedlag}
The Alternating Directions Method of Multipliers is directly applicable to the form of the convex optimization problem~\eqref{ConvObj}--\eqref{ConvConst}. To do so, we write the augmented Lagrangian function of the problem, for a given penalty parameter $\lambda^{-1} >0$, where $\u$ is the vector of Lagrange multipliers associated to the constraints~\eqref{ConvConst}:

\begin{equation}
L_{\lambda^{-1}}(\x,\y,\u) = g^\alpha(\y) +\iota(\x) + \u^\T(\y-\B\x) + \frac{1}{2\lambda} ||\y-\B\x||^2
\end{equation}
where for two vectors $\a$ and $\b$, $\a^\T \b$ is the Euclidean product of $\a$ and $\b$ and $||\cdot||$ is the Euclidean norm.

Unlike the classic method of multipliers (see for instance \cite{bertsekas2014constrained}, Chapter 3), where the augmented Lagrangian is minimized jointly with respect to $(\x,\y)$ each time the dual variable $\u$ is updated (see Eq.~\eqref{u_up} below), the ADMM consists of minimizing the augmented Lagrangian $L_{\lambda^{-1}}(\x,\y,\u)$, alternatively with respect to $\x$ and $\y$. After some straightforward algebra, one can derive the update rules that represent each of the alternate minimizations, and summarize, through Eqs.~\eqref{y_up}--\eqref{u_up} below, one iteration of ADMM that is to be repeated till a suitable termination condition is satisfied. 

\begin{align}
\y \longleftarrow &  \argmin L_{\lambda^{-1}}(\x,\cdot,\u)\nonumber\\
 & =\argmin_{\y} g^\alpha(\y) + \frac{1}{2\lambda} ||\y-(\B\x - \lambda \u)||^2 \label{y_up}\stepcounter{equation}\tag{\theequation a}\\ 
\x \longleftarrow & \argmin L_{\lambda^{-1}}(\cdot,\y,\u) \nonumber\\ 
 & = \argmin_{\x} \iota(\x)  + \frac{1}{2\lambda} ||\B\x - (\y + \lambda \u)||^2 \tag{\theequation b} \label{x_up}\\ 
\u  \longleftarrow & \u + \frac{1}{\lambda}(\y - \B\x)\label{u_up} \tag{\theequation c}
\end{align}
We refer to this algorithm as the \emph{centralized} algorithm.

In the light of the above update rules, one can directly remark that the centralized algorithm assumes that the function $\iota$ is globally accessible (see Eq.~\eqref{x_up}). The function $\iota$ is the indicator function of the network's capacity set, hence contains and necessitates full knowledge of the network topology and load. However, we assume that our distributed setting cannot afford such global knowledge. Indeed, the SDN controllers have the sole topology information of the underlying sub-network they have been assigned to. Thus, they are able to perform local computations only and have to operate together to achieve global consensus, which forbids the function $\iota$ from being globally available. 
\begin{remark}
The centralized algorithm~\eqref{y_up}--\eqref{u_up} benefits from the best presently known convergence rate. Indeed, formulation~\eqref{ConvObj}--\eqref{ConvConst} belongs to a class of problems with strongly convex objectives having Lipschitz gradient that respects the assumptions of Theorem~1 in \cite{Deng2016} and as such, the convergence of ADMM applied to solve them is linear with a known explicit rate. Unavoidably, we will sacrifice this convergence speed guarantee in order to abide by the rules of the distributed SDN settings.
\label{remark:linear}
\end{remark}
In the next section, we show how to break down the topology information with respect to the SDN distribution, by benefiting from the distributive properties of ADMM.

\section{\Add{Multi-agent consensus formulation}}
\label{sec:consensus}

\Add{In this section, we show how to skirt the need for the global topology information by decomposing the formulation with respect to the network links of each SDN domain\Add{.}  The global knowledge of the topology being not affordable in the distributed SDN control plane, the decomposition permits to respect the locality of information of the different domain controllers.} As we will see, our decomposition will break down the function $\iota$ by partitioning the network topology into different domains, and will naturally induce a partition of the set of requests over the network. This will therefore partition the global information and distribute it among the domains accordingly. We assume the domains can operate only with their private information, and any other information needed that belong to another controller needs to be gathered through inter-domain communication. The partition into domains can be orchestrated at the discretion of the SDN architect. Unavoidably though, domains will need to exchange information as paths may traverse many of them. 

\subsection{Preliminaries}
\label{prel}
We suppose that the network is split into $M$ domains, where each domain $m$ is assigned to a set of links $\setJ_m \subset \setJ$. In other words, $(\setJ_m)$ forms a \emph{partition} of the set of links $\setJ$.
Let $\setP_m$ be the set of paths traversing the domain $\setJ_m$ via some link $j\in \setJ_m$. 

We partition the set of requests $\setR$ in the following way. Each request $r$ has one source node that belongs to one unique domain $m(r)$. Define, for each domain $m$, the set of requests originating from it as $\setR_m = \{r\in \setR \mbox{ s.t. } m = m(r)\}$.  Notice that $(\setP_m)$ forms a covering of $\setP$, and $(\setR_m)$ forms a partition of $\setR$.

Now, each domain $m$ can define a set of (private, as explained later) indicator functions for each of its links: let $\iota_j$ denote the indicator function for link $j \in \setJ_m$, i.e.,

\begin{equation}
\iota_j(\x) =
\left\{ \begin{array}{ll}
0 & \mbox{if} \displaystyle \sum_{p: j\in J_p} x_p \leq c_j  \mbox{ and } x_p \geq 0 \mbox{ } \forall p \mbox{ s.t. } j\in J_p \\
\infty & \mbox{otherwise}. 
\end{array}\right.
\end{equation}
\begin{remark} \label{rem:iota}
The function $\iota_j$ depends only on the variables $x_p$ such that $j\in J_p$.
\end{remark}
%




To simplify the algorithm design, we eliminate the matrix $\B$ and the variable $\y$ from the formulation by plugging in directly the equation~\eqref{aggregate_eq} into the functions $g_r^\alpha$. We therefore redefine the (now private) domain functions $g^\alpha_m$ as:

\begin{equation}
g_{m}^\alpha: \R_+^{\setP} \to \R, \x \mapsto \sum_{r\in \setR_m} g^\alpha_r(\sum_{p\in P_r} x_p) ,
\end{equation} 
\begin{remark} \label{rem:g}
In fact, the function $g_{m}^\alpha$ depends only on the variables $x_p$, $p\in P_r$, for $r\in \setR_m$.
\end{remark}

\subsection{Private variables and consensus form}

As from now, the controller of domain $m$ is the sole owner of the now private indicator functions $\iota_j, j\in \setJ_m$ and $g_m^alpha$. In order to evaluate its private function $g_{m}^\alpha$, domain $m$ defines a variable $\x^m$ as its private copy of the variable $\x$. As $g_{m}^\alpha$ depends only on a reduced number of variables (see Rem.~\ref{rem:g}), one only needs to define $\x^m = \left((x^m_p)_{p\in P_r}\right)_{r\in \setR_m}$.

In the same way, each link $j$ belonging to domain $m$ is associated to a private copy  $\z^{j}$ of the vector $\x$ through which the controller will evaluate the indicator function $\iota_j$. Likewise, thanks to Rem.~\ref{rem:iota}, we only need to define $\z^{j} = (z^{j}_p)_{p: j\in J_p}$, to avoid creating unnecessary variables. 

The notation $\z^m$ will from now on be adopted to refer to the collection of variables $\{\z^{j}\}_{j\in \setJ_m}$.

In short, the minimal information that a domain $m$ needs is:
\begin{enumerate}
\item the SDN-domain topology $\setJ_m$ and the load on each link $j\in \setJ_m$, that is, the sub-matrix $(A_j)_{j\in \setJ_m}$.  
\item the knowledge of the objective functions $g^\alpha_r$ for each $r \in \setR_m$.
\end{enumerate}

Following the above notations, we can reformulate the objective~\eqref{ConvObj} as the following:

\begin{equation}
\displaystyle \sum_{m = 1}^M  h^m(\x^m, \z^m),
\end{equation}

where the functions $h^m$ are defined as:
\[\displaystyle h^m(\x^m,\z^m) \coloneqq  g_{m}^\alpha(\x^m) + \sum_{j\in \setJ_m} \iota_j(\z^{j}).\]
Finally, the controllers will enforce a consensus value among all the copies of a same variable via the following (global) indicator function:

\begin{equation}
\label{chi_consensus}
 \chi(({\x}^m)_m, ({\z}^{m})_m) = \left\{\begin{array}{ll} 0 & \mbox{if } z^{j}_p - x^m_p = 0 \\ &  \forall m, \forall p\in P_r, \forall r\in \setR_m, \forall j\in J_p\\ \infty &\mbox{otherwise.} \end{array}\right.
\end{equation}

Hence, the new function to minimize to solve our problem is simply $\sum h^m +\chi$. We separate it further to end up with the classic 2-block form applicable to ADMM by introducing special copies ${\x'}^m$ and ${\z'}^m$ of the collections of variables $\x^m$ and $\z^m$, respectively for the composite term $\chi$. Considering the present setting, one can verify that our problem takes the equivalent separated form:
%

\begin{align}
\min & \displaystyle\quad \sum_{m = 1}^M   h^m(\x^m, \z^m) + \chi(({\x'}^m)_m, ({\z'}^m)_m)\label{final_objective}\\
\mbox{s.t.   } &\qquad \z^j = {\z'}^j \quad \forall j \in \setJ \label{f_consensus} \\
& \qquad \x^m = {\x}'^m \quad \forall m =1,\ldots, M.\label{final_consensus}
\end{align}

To recap, \emph{we have artificially separated the objective function by creating a minimal number of copies of the variable} $\x$ \emph{in order to fully distribute the problem}. Now, instead of a global resource allocation variable, several copies of the variable account for how its value is perceived by each link of each domain. This new formulation can be interpreted as a multi-agent consensus problem formulation where domain $m$ has cost $h^m$ and has to agree on the values it shares with other domains. To enforce an intra- (local) and inter- (global) domain consistent value of the appropriate allocation, consensus constraints \eqref{chi_consensus} are added to the problem. Those consensus constraints are now handled by the agent $\chi$ that will represent the communication steps required at each iteration between domains. 
  
%
%

\subsection{Update rules}
In the same fashion as in Section~\ref{augmentedlag}, we write the augmented Lagrangian form of Formulation~\eqref{final_objective}--\eqref{final_consensus}, and obtain, after some simplification, Algorithm~\ref{FDADMM} (Fast Distributed (FD)-ADMM). The variables $\u^m$ and $\v^m$ are the dual variables associated to the constraints~\eqref{f_consensus} and \eqref{final_consensus} respectively. \\
Alg.~\ref{FDADMM} can be summarized into the three following steps:
\begin{description}
\item[ \bf Stage 1)]\hspace{13pt}Objective minimization -- solve :
 \begin{align} 
  & \argmin_{(\x,\z)}\sum_{m = 1}^M \Bigg\{  h^m(\x^m, \z^m)+ \nonumber \\ 
  &  \frac{1}{2\lambda} \left(||\x^m-{\x'}^m + \lambda \v^m||^2 +||\z^m-{\z'}^m + \lambda \u^m||^2\right) \Bigg\}\label{stage1} 
 \end{align}
\item[\bf Stage 2)]\hspace{13pt}Consensus operation -- setting $\x^m$ and $\z^m$ as the result of Stage 1), solve :
\begin{align}
 & \argmin_{(\x',\z')}\chi(({\x'}^m)_m, ({\z'}^m)_m)+ \nonumber\\
 & \frac{1}{2\lambda} \sum_{m = 1}^M \left(||\x^m-{\x'}^m + \lambda \v^m||^2 +||\z^m-{\z'}^m + \lambda \u^m||^2\right) \label{stage2}
\end{align}
\item[\bf Stage 3)]\hspace{13pt}Dual variables update -- see Alg.~\ref{FDADMM} lines~\ref{dualj} and \ref{dualx}.
\end{description}

For presentation purposes we presented those three stages in Alg.~\ref{FDADMM} in the order 2), 3), 1)  to start with the information gathering and end with the information broadcasting of each domain controller. Of course, it has no incidence on the algorithm behavior as long as the three steps are executed in the correct order. 
%

By separability of Eq.~\eqref{stage1}, the objective minimization stage boils down to the parallel minimization of each term of the sum by the corresponding domain controller. In details, the controller of domain $m$ can minimize the term of index $m$ in the sum \eqref{stage1} with respect to $\x^m$ and $\z^m$, \emph{simultaneously}, meaning that the minimization operations with respect to $\x^m$ and $\z^m$ are \emph{independent}. 

The corresponding update rules can be stated as follows:

{\noindent\bf Stage 1) a) The minimization over $\x^m$:}  \[\displaystyle\argmin_{\x^m} \sum_{r\in \setR_m}\left\{ g^\alpha_m(\x^m) + \frac{1}{2\lambda} \sum_{p\in P_r}||x_p^m - {x_p'}^m +\lambda v_p^m||^2\right\}.\]
This problem is separable with respect to $r\in \setR_m$. Hence, we minimize it term by term. A term of index $r$ of the above sum is differentiable, strictly convex (because of the 2-norm term) and coercive (near $0$ and $\infty$) with respect to the positive variables $(x_p^m)_{p\in P_r}$. Therefore, a unique minimum exists and is obtained at the unique critical point. By differentiation, the critical point verifies:
\begin{equation*}
x^m_p - {x_p'}^m + \lambda v^m_p = \frac{\lambda w_r}{\displaystyle(\sum_{r\in P_r} x^m_p)^\alpha}.
\end{equation*}
Thus, by setting $y_r \coloneqq \sum_{r\in P_r} x^m_p$, one has:
\begin{equation}
y_r^\alpha (x_p^m -{x_p'}^m + \lambda v^m_p) = \lambda w_r.
\label{eq:intermediary}
\end{equation}
Summing Eq.~\eqref{eq:intermediary} over $p\in P_r$, a necessary condition for $y_r$ is:
\begin{equation}
y_r^{\alpha +1} - y_r^\alpha\sum_{p\in P_r} \{x_p' - \lambda v_p^m\} -\lambda w_r = 0.
\label{eq:proximalyr}
\end{equation}
A quick study of the left hand side of Eq.~\eqref{eq:proximalyr} shows that there is a unique positive solution. Therefore, it is an uni-variate  unconstrained problem that can be solved efficiently with standard root finding algorithms (particularly when $\alpha$ is an integer, it boils down to finding the unique positive root of a polynomial of degree $\alpha+1$).
Eq.~\eqref{eq:intermediary} in turn yields the path-wise update rule:
\begin{equation}
x^m_p = \frac{\lambda w_r}{y_r^\alpha} + {x_p'}^m-  \lambda v_p^m.
\label{eq:xupdatefinal}
\end{equation}

{\noindent\bf Stage 1) b) The minimization over $\z^m$:}
 \[\displaystyle\argmin_{\z^m} \sum_{j\in \setJ_m}\left\{ \iota_j(\z^j) + \frac{1}{2\lambda} \sum_{p: j\in J_p}||z_p^j - {z_p'}^j +\lambda u_p^j||^2\right\}.\]
Likewise, this problem is separable with respect to $j\in \setJ_m$. Hence, we minimize it term by term. The minimization of the term $j$ boils down to operating the Euclidean projection of the point $\phi^j =({z_p'}^j -\lambda u_p^j)_{p: j\in J_p}$ onto the capacity set of the link $j$. That is, we find the closest point\footnote{With respect to the 2-norm.} to $\phi^j$ lying in the set $\{(z_p^j)_{p:j\in J_p} \geq 0 \mbox{ s.t. } \sum_p z_p^j \leq c_j\}$. This set is in fact a simplex of dimension $n_j =\mathrm{Card}(\{p:j\in J_p\})$ and radius $c_j$, and this operation can be done \cite{chen2011projection} with a complexity dominated by the one of sorting a list of length $n_j$ (thus, $\mathcal{O}(n_j \log(n_j))$ in average).


\noindent{\bf Stage 2) The consensus operation:}
As for the  consensus operation stage, it suffices\footnote{The minimization of Eq.~\eqref{stage2} can be done with a straightforward calculus that we ommit here. Moreover, we explain in \cite{allybokus2017real} how to obtain this simplified expression.} to build the consensus point $\tilde{\bm{z}} \in \R^{\setP}$, where:
\begin{equation}\label{eq:average}
\tilde{z}_p = \frac{1}{1+ |J_p|}\left(x^{m(r)}_p + \sum_{\tiny{\begin{matrix}m =1 \ldots M\\ j\in \setJ_m\cap J_p\end{matrix}}} z^{j}_p\right) \quad \forall r, \forall p\in P_r. 
\end{equation}
  
 
In order to produce this point, domain $m$ communicates its contribution $\dot{\z}^m$ to the sum encountered in the average~\eqref{eq:average} (l.\ref{alg2:comm}). Then the actual consensus value for each component $p$ is recovered by dividing by the number of existing copies of variables with index $p$. Let $r$ be the request such that $p\in P_r$. Then, one can check that for each path $p$, this number equals $|J_p|+1$: one copy per link and one copy for the domain $m$ such that $m(r) = r$ (l.\ref{form}).

 \begin{algorithm}[t!]
 \caption{Fast Distributed ADMM (FD-ADMM)}\label{FDADMM}
 \begin{algorithmic}[1]
 \Procedure{}{}$\mbox{\textbf {of domain} } m${}
 \\\hrulefill
 \Statex {\color{blue} \sc Stage 2):}
 \State Collect $\dot{\bm{z}}^q$ from neighbors \label{receive}
 \State \mbox{\sc Form} $\tilde{\bm{z}}^m_{p} = \frac{1}{|J_p|+1} \left(\sum_{n: p\in \setP_n} \dot{\z}_{p}^{n}\right)$ \Comment{$\forall r\in \setR_m, p\in P_r$} \label{form}
\\\hrulefill
 \Statex {\color{blue} \sc Stages 3) and 1):}
 	\State $ u^{j}_{p} \gets u^{j}_{p} +\frac{1}{\lambda}(z^{j}_{p} -\tilde{{z}}^m_{p})$\Comment{$\forall p \mbox{ s.t. } j\in J_p$} \label{dualj}
 	\State Update $\z^j$ as the Euclidean projection of 
 	\Statex $(\tilde{z}^m_p - \lambda u^j_p)_{p: j\in J_p}$ onto the capacity set of 
 	\Statex link $j$\Comment{$\forall j\in \setJ_m$} 
 	\Statex (see {\bf 1. b)}
 	
\\
\State $v^m_{p} \gets v^m_{p} + \frac{1}{\lambda}(x^m_{p} - \tilde{z}^m_{p})$ \label{dualx}
\State Update $x^m_p$ as in Eq.~\eqref{eq:xupdatefinal} \Comment{$\forall r\in \setR_m, p\in P_r$}
\\\hrulefill
\Statex {\color{blue}\sc Build information to be sent to neighbors:}
\State ${\dot{z}}^{m}_p \gets	\displaystyle\sum_{j\in J_p  \cap \setJ_m} z^j_{p}$\Comment{$\forall p\in \setP_m$}\label{begincom}\label{alg2:comm}
\State $ {\dot{z}}^{m}_p \gets {\dot{z}}^{m}_p + x^m_p$ \Comment{$\forall r\in \setR_m$} 
\State Send ${\dot{\bm{z}}^m} $ to neighbors\label{send}

\EndProcedure
 \end{algorithmic}
 \end{algorithm}

\noindent \textbf{Communication among domain controllers:} In FD-ADMM, \emph{only domains that do share a path together have to communicate}. The communication procedures among the domain controllers are described between the lines~\ref{begincom} and \ref{send}. In these steps, the domains gather from and broadcast to adjacent domains the sole information related to paths that they have in common. In particular, domains are blind to paths that do not visit them, and can keep their internal paths secret from others. In details, after each iteration of the algorithm, each domain $m$ receives the minimal information from other domains such that $m$ is still able to compute a local consensus value $\tilde{z}^m$ (l.~\ref{form}). Next, domain $m$ sends back to the neighboring domains a contribution $\dot{z}^m$ so that they can recover a consensus value of the path-wise allocation.

\emph{Communication overhead:} In terms of overhead, we can easily evaluate the number of floats transmitted between each pair of domain at each iteration. At each communication, domain $m$ must transmit $\dot{z}^m_p$ for each path $p\in \setP_m$ to each other domain that $p$ traverses. The variable $\tilde{z}$ does not need to be centralized or transmitted between controllers. Each domain controller may actually have a copy of it recovered locally at negligible cost (see l.\ref{form}). Note that the value of $|J_p|$ in l.\ref{form} can be recovered by each domain through a \emph{unique} message passing at the establishment of the path $p$, therefore we can dismiss the global nature of this information. Hence, domain $m$ transmits in total $\sum_{n \neq m} |\setP_n\cap \setP_m|$ floats to the set of its peers. As a comparison, in a distributed implementation of the standard dual algorithm given in \cite{mccormick2014real}, each domain $m$ would transmit in total $\sum_{n\neq m} |\{j\in \setJ_m, \exists p\in P_n \mbox{ s.t. }  j\in J_p\}|$ floats to the set of its peers, which is bounded by $(M-1)|\setJ_m|$ as $|\setP|$ grows. 

\Add{\noindent \textbf{Practical implementation.} Domain controllers implementing FD-ADMM are communicating bandwidth allocation decisions to SDN switches using standard protocols such as OpenFlow, PCEP (Path Computation Element Protocol) or BGP-LS (Border Gateway Protocol for Link State) available at their south-bound interface. To exchange information with other domain controllers about optimization variables or network states, i.e. east-west communications, domain controllers may use the iSDNi or IOCONA interfaces, in the case of the two most popular open source SDN controllers, OpenDayLight \cite{opendaylight} and ONOS \cite{onosproject} respectively.}

\noindent \textbf{Feasibility preservation:} A potential drawback of the distributed approach is the potential feasibility violation by the iterate $\tilde{z}$. However, we have the following positive result.
\begin{claim}
\label{claim:feasible}
At each iteration of FD-ADMM, the point $\z^\dagger$, defined as $z^\dagger_{p} = \min_{j\in J_p} z^{j}_{p}, \mbox{  } \forall p\in \setP$, is feasible, and $\B\z^\dagger$ converges to the optimal aggregate bandwidth allocation. 
\end{claim}
\begin{proof}
See \cite{allybokus2017real}, Proposition 1.
\end{proof}

\noindent 
Thus, in a certain way, for sufficiently loaded and communicating domains (i.e. the $|\setP_m\cap \setP_n|$ are large enough) we sacrifice some overhead (counted on a per iteration basis) compared to standard dual methods, but in exchange for anytime feasibility, a major feature that dual methods do not generically provide.

\noindent {\bf{Penalty parameter initialization:}} It is well known that the penalty parameter $\lambda$ highly conditions the convergence speed of ADMM. In \cite{allybokus2017lower}, it has been shown numerically that the optimal penalty $\lambda_*$ in terms of convergence speed for the (centralized) algorithm given in the update rules~\eqref{y_up}--\eqref{u_up} provides a satisfactory performance of FD-ADMM. The analysis of this aspect of the design of FD-ADMM goes beyond the scope of this article, but in this paragraph, we adapt a multi-path version of Theorem~1 in \cite{allybokus2017lower} where a lower bound on the (path-wise) $\alpha$-fair bandwidth allocation in the case of single paths is derived. It is shown there how this lower bound permits to define a natural penalty parameter initialization that boosts the algorithm performance. Here, we generalize the result of \cite[Th. 1]{allybokus2017lower} to our case of multi-path routing, and establish a lower bound on the \emph{aggregate bandwidth} allocation.  

\Add{To do so, one needs the following definitions: we define, for each request $r$, the set $\setR(r)$ of \emph{other requests that visit some link used by $r$}: $\setR(r) \coloneqq \left\{ s\in \setR: \left( \cup_{q\in P_s} J_q\right) \cap \left( \cup_{p\in P_r} J_p\right) \neq  \emptyset \right\}$. Lastly, we define the \emph{utopic}\footnote{In fact, the utopic path-wise and aggregate bandwidth allocations correspond to the path-wise and the aggregate bandwidth allocations, respectively, that a request can get if all other requests get an allocation of 0, which justifies the term \emph{utopic}.} 
aggregate bandwidth allocation $a_r$ of a request $r$, and its \emph{local midpoint value} $\varrho_r$, respectively, as the following: 
\begin{equation}
a_r \coloneqq \max_{{\scriptsize{\begin{matrix}
x \geq 0, \A\x \leq \c, \\
\forall s\in \setR(r) - \{r\}, (\B\x)_s= 0
\end{matrix}}}} (\B\x)_r, \quad \varrho_r = \frac{w_r}{\displaystyle \sum_{s\in \setR(r)} w_s} a_r.
 \end{equation}}
 \vspace{-6mm}
\begin{claim}
Let $r_\dagger \coloneqq \argmin_{s\in \setR} \varrho_s$ be the request with the smallest local midpoint value. Then, one can bound from below the optimal \emph{aggregate bandwidth} allocation $\y^* \geq \d$ where:
\begin{itemize}
\item if $\alpha \geq 1 \quad  d_r =\varrho_{r_\dagger} ^{1-1/\alpha} \varrho_r ^{1/\alpha}$
\item if $0<\alpha \leq 1 \quad d_r =\left(\displaystyle  \frac{w_r a_r}{\displaystyle\sum_{s\in \setR(r)} w_s a_s^{1-\alpha}}\right)^{1/\alpha}.$
\end{itemize}
\end{claim}
\begin{proof}
See \cite{allybokus2017lower}, Theorem 1.
\end{proof}
Using this lower bound, we now generalize the penalty parameter that is formulated in \cite{allybokus2017lower}:

\vspace{2mm}
\noindent \textit{Penalty parameter initialization:} \begin{equation} \lambda_* = \alpha \left( \min \frac{w_r}{a_r^{\alpha+1}} \max \frac{w_r}{d_r^{\alpha+1}}\right)^{-\frac{1}{2}}.\end{equation}

\section{Reconfiguration with switching costs}
\label{sec:reconfiguration}

In this section, we show that our FD-ADMM formulation is flexible enough to permit to solve a related relevant problem, by simply modifying the private objective functions $h^m$. 
We assume a traffic is already established with a current resource allocation, and that its requirements can vary on-the-fly. One can model a variation of the traffic requirements by a change in priorities between flows via a variation of the weight vector $\w$, the computation of a new path for an existing (or not) request, the elimination of a path for a request, etc. Under these circumstances, FD-ADMM can continuously generate feasible solutions to adapt the path-wise allocation to the new requirements in real-time. In fact, by doing so, the controllers may improve the optimality gap, and thus satisfy the demands with a better fairness measure as they evolve. However, enforcing a new resource allocation too often requires overwhelming flow reconfigurations rules that can cause Quality-of-Service degradation or system instability \cite{paris2016controlling}. Therefore, we introduce a switching cost to limit the number of reconfigurations. The goal for the controllers will thus be to perform a trade-off between fairness and switching cost. 

The introduction of a switching cost into the objective function can be of interest to enforce hard constraints onto the number of reconfigured paths. To be more specific, let $\x^0$ be a feasible path-wise allocation and assume the actual resource allocation of the demands follows $\x^0$. Now, the traffic demands have changed and the network has to recompute a new path-wise allocation $\x^*$ with fair aggregate bandwidth allocation $\B\x^*$, to respond to the traffic requirements. Assume the network has a budget of  $\kappa >0$ reconfigurations. According to the fairness policy of the network, the allocation should be updated in order to maximize the new fairness metric, without exceeding this budget: 
\begin{equation}||\x^0 - \x^*||_0 \leq \kappa, \Leftrightarrow \sum_p {\1({x^0_p \neq x^*_p})} \leq \kappa,
\label{sparse_constraint}
\end{equation} 
where $||u||_0 = \mathrm{Card}\{p, u_p \neq 0\}$ is called the \emph{zero-norm}\footnote{This is an abuse of terminology as it is not a norm.} of a vector and denoted $\ell_0$. Adding the constraint of Eq.~\eqref{sparse_constraint} into the problem gives rise to a problem structure with integral constraints, and goes beyond the scope of this work. We consider here a relaxation of this problem that is still tractable with our method.  

We can control the zero-norm~\eqref{sparse_constraint} by adding the most natural sparsity inducing penalty induced by the $\theta$-scaled $\ell_1$-norm $\theta ||\x^0 - \x^*||_1$, where $\theta$ is a positive parameter. The $\ell_1$-norm is well known to be the fittest convex relaxation of the $\ell_0$-norm, for the simple reason that the $\ell_1$-ball is the convex hull of the set of points $\{\v$ s.t. $||\v||_0 \leq 1 \}$. We therefore consider the problem described in Eqs.~\eqref{final_objective}, \eqref{f_consensus}, \eqref{final_consensus}, with an extended expression of the function $h^m$:

\begin{align}
h^m(\x^m,\z^m) =  g_{m}^\alpha(\x^m)+\theta \sum_{r\in \setR_m}\sum_{p\in \setP_r} |x^m_p - x^{m0}_p| \nonumber \\ + \sum_{j\in\setJ_m} \iota_j(\z^m), 
\end{align}
where $\x^{m0}$ is the copy of $\x^0$ for domain controller $m$.

With this extended formulation, the changes in FD-ADMM only occur in the optimization stage 1) (Eq.\eqref{stage1}). The term-wise minimization of the functions $h^m$ now also takes into account an incentive for each domain $m$ to stay near the point $\x^{m0}$. Of course, a proper tuning of the parameter $\theta$ is necessary to enforce the real budget $\kappa$. The larger $\theta$, the smaller the number of re-sized paths. We show this effect in the next section, dedicated to the experimentations.

\begin{figure*}[t]
\includegraphics[width = \textwidth]{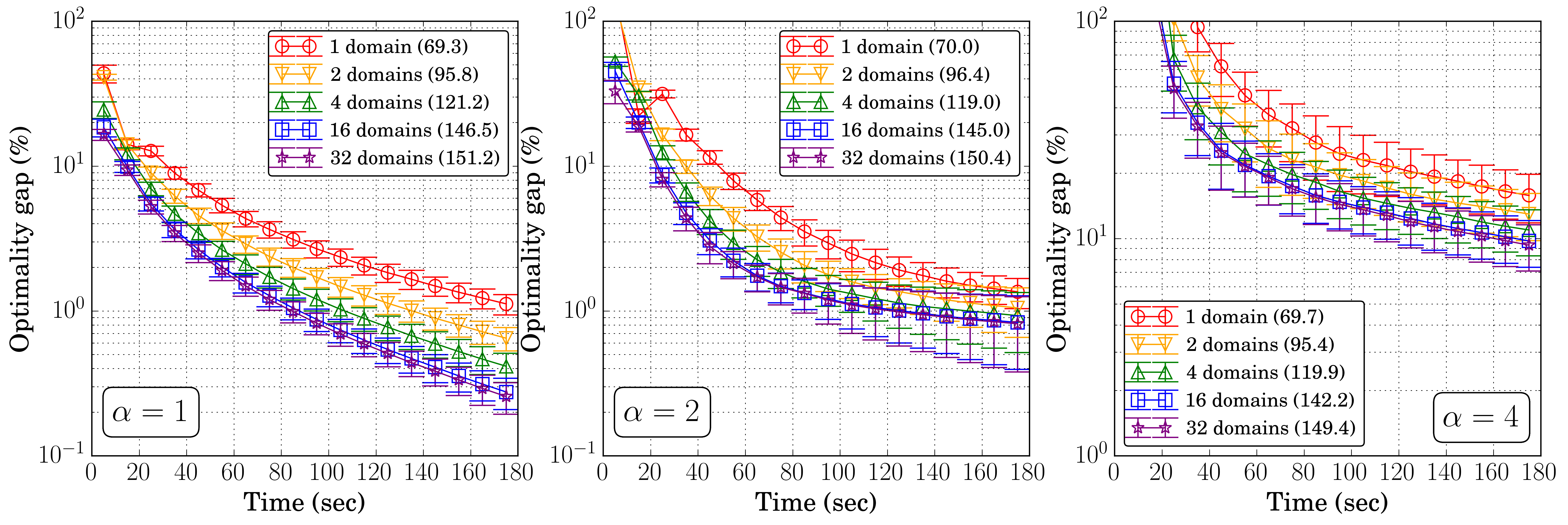}
\caption{{\bf Convergence with domain distribution.} (Barabasi-Albert Graphs) Optimality gap over time within a deadline of 3 minutes for different numbers of domains. The average number of achieved iterations appear in the legend (inside parenthesis). }
\captionsetup{justification=centering}
\label{Figure1}
\end{figure*}
\section{Simulations}
\label{sec:execution}

This section is dedicated to the experimentation of FD-ADMM. First of all, we demonstrate, on large instances with hundreds of nodes and thousands of requests and paths, the gains achievable with the distribution of the workload among several domain controllers (Fig.~\ref{Figure1}). We will not focus on absolute performance evaluation in terms of time, but in demonstrating the benefits of the distribution of the SDN centralized controller into SDN domain controllers that split the workload and build global solutions through consensus. Secondly, we analyze the behavior of FD-ADMM implemented as the extension described in Sec.\ref{sec:reconfiguration} (Fig.~\ref{Figure3}). All the simulations are executed under the three main sharing policies: proportional fairness ($\alpha = 1$), minimum potential delay ($\alpha = 2$) and max-min fairness (with an approximation $\alpha = 4$).
\Add{In previous works \cite{allybokus2017real}, the performance of FD-ADMM with respect to convergence speed, feasibility preservation, and responsiveness in variable traffic requirement situations was extensively compared with the one of the classic Lagrangian method in \cite{kelly1998rate}. We do not display these results here, and refer the curious reader to the aforementioned work.} 
\vspace{-5mm}

\subsection{Setting}
\Add{Two simulations\footnote{The curious reader can find the source code an run the simulations in \cite{sourcecode}. } were run over two types of networks. The one type of network was generated following the model of Barabasi-Albert with minimal degree $4$, and the other was a Fat Tree with a number pods $k = 16$. We give here a description of the generated instances for the first simulation.

\emph{The Barabasi-Albert networks} contained $500$ nodes, which gave problems with $3968$ resources. Over this network, instances of requests were created between sources and destinations chosen uniformly at random with a number of established paths from $2$ to $4$ between the pair of nodes. The instances contained $5,000$ requests (that represented approximately $11,000 - 15,000$ paths in total).

\emph{The Fat Tree (with 16 pods) networks} contained $1,345$ nodes, which gave problems with $3,136$ resources. We modeled a connection to the Internet via the fat tree by adding a root node connected to the core nodes of the tree. For each server, we generated two connections to the root node through $4$ paths each, and to another server chosen uniformly at random through $4$ paths. This gave problems with $3,072$ requests and therefore $12,288$ paths.}

The second simulation was executed on a small network \Add{(generated with the model of Barabasi-Albert with same minimal degree $4$)} with $100$ nodes made of $768$ links, and each instance contained $500$ requests each with $1$ to $2$ established paths (approximately $700 - 800$ paths). 

\Add{In all the simulations,} the network capacities were fixed at an equal value ($100$) and unless specifically mentioned otherwise, the weights $w$ were fixed at a unique value $1$. We created the domains by splitting the set of network links into a number (equal to the desired number of domains) of equally sized subsets, taking into account the network topology so that each domain remains connected. We assumed all the domains performed their update rule in a synchronized manner, meaning that an iteration was achieved (and the variable $\tilde{\z}$ in stage 2) updated) \emph{when all the controllers were done with their respective work}. The direct extension of distributed ADMM-based algorithms to the asynchronous setting is possible and its convergence is studied and demonstrated for instance in \cite{chang2016asynchronous}. We keep simulations in the aforementioned case for future work.


\subsection{Results}

\begin{figure*}[t]
\includegraphics[width = \textwidth]{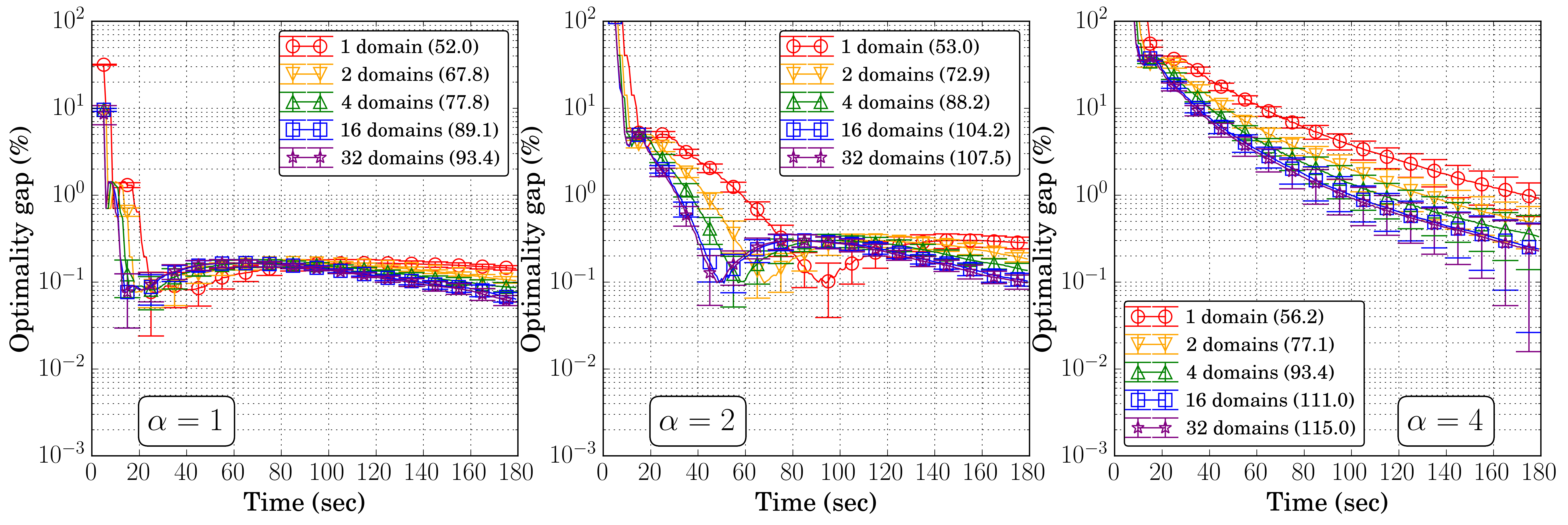}
\caption{\Add{{\bf Convergence with domain distribution.} (Fat Tree Graphs) Optimality gap over time within a deadline of 3 minutes for different numbers of domains. The average number of achieved iterations appear in the legend (inside parenthesis). }}
\captionsetup{justification=centering}
\label{Figure1tree}
\end{figure*}

\noindent{\bf Convergence with domain distribution.} We evaluated the gains achievable by the SDN distribution by plotting, for each partition of the network following Sec.~\ref{sec:consensus} into a number of controllers within $\{1,2,4,16,32\}$, the achieved number of iteration and the optimality gap\footnote{The optimal solution was obtained by running FD-ADMM till the convergence is provably obtained. This can be done by observing the values of the primal and dual residuals of the problem at each iteration. It is shown \cite[Ch. 3.3.1]{boyd2011distributed} that these values bound the optimality gap of the iterate at each iteration. The residual values are thus used as a robust convergence detector. We consider the optimum was reached when the residual values dropped below $10^{-2}$ (modest but satisfactory convergence)} with time under a finite time deadline. In our implementation, each controller performs its private update rules under one specific thread at each iteration. As the update rules for each controller are also massively parallelized, there is still a lot of room for optimizing absolute time performances. Also, the consensus operation stage is here done without exploiting the parallelism. In practice, domain controllers may be themselves equipped with hundreds of cores and perform multi-threads computations that can crush the computation times down to several orders of magnitude smaller. \Add{The deadline was fixed to three minutes for all instances.} The results are shown in Fig.~\ref{Figure1} \Add{(Barabasi-Albert Graphs) and Fig.~\ref{Figure1tree} (Fat Tree Graphs)}. Each point represented corresponds to an average over $15$ and $10$ instances \Add{(for Fig.~\ref{Figure1} and Fig.~\ref{Figure1tree}, respectively)} of the same characteristics along with a $95\%$-confidence interval estimated following the $t$-distribution model. 

The results show that the distribution of the computation permits to operate more than twice faster (according to the average achieved number of iterations) and thus to reach the same optimality gap faster. For instance, for $\alpha = 1$ \Add{in Fig.~\ref{Figure1}}, in order to reach a gap below $2\%$, a single controllers needs more than 2 minutes whereas $32$ controllers together need less than one minute. This reduction of the time can be even more dramatic when the update rules within each domain, as well as the consensus operation stage, are done in parallel. 

The communication delay being dependent of the SDN implementation, we condense the information concerning the overall communication delay of controllers into the number of performed iterations. According to recent work on this topic \cite{benamrane2017new}, east-west interfaces between domain controllers can permit one to perform inter-controller communications with low latency. With the fast convergence that is demonstrated in Fig.~\ref{Figure1} in terms of iteration count, one can see that the communication delay does not deteriorate the acceleration achieved with the distribution of the workload among the domain controllers. 
\vspace{2mm}

\begin{figure*}[h!]
\includegraphics[width = \textwidth]{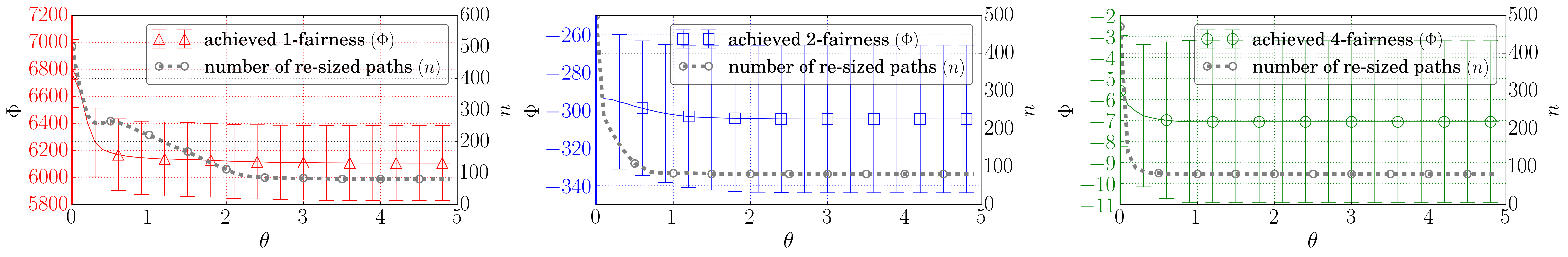}
\centering \caption{{\bf Switching costs.} Number of re-sized paths ($n$) and achieved fairness ($\Phi$) versus the switching cost $\theta$.}
\captionsetup{justification=centering}
\label{Figure3}
\end{figure*}

\noindent{\bf Switching costs.} In the second experiment, we set $x^0$ as an optimal path-wise allocation for the actual setting of the weights $\w (\equiv \1)$. Then, we simulate a new traffic requirement by choosing at random a new weight vector $\w_1$ within $[1,10]$. 
Thus, FD-ADMM runs to find the new optimally fair aggregate bandwidth allocation, but, the value of $\theta$ forces a trade-off between fairness and switching cost. We ran FD-ADMM till convergence was provably obtained to a precision $10^{-2}$ and analyzed the number of paths that had been re-sized, for different values of $\theta$. The minimization problem in stage 1) was solved optimally at each iteration using a standard optimization package for unconstrained problems with a tolerance $10^{-2}$ (as the closed form of \eqref{eq:xupdatefinal} is no longer available for positive $\theta$). To evaluate numerically the zero-norm of vectors, we used the function $N_\epsilon(x)=  \mathrm{Card}(\{p: |x_p| >\epsilon\})$ where $\epsilon$ is a desired level of precision. This precision was fixed to the convergence tolerance of our experimentation, that is, $\epsilon = 10^{-2}$. At the optimum, the objective therefore splits up into a sum $-\Phi +\theta \Psi $, where $\Psi \coloneqq  \sum_{p, |x_p - x^0_p| > \epsilon} |x_p - x^0_p|$, the number of terms in that sum being equal to $n\coloneqq N_\epsilon (\x - \x^0)$ and $\Phi$ is the achieved fairness $\sum_r{f^\alpha_r(y_r)}$ for the optimal aggregate bandwidth allocation $\y = \B\x$. The results are shown in Fig.~\ref{Figure3}. Likewise, each point represented corresponds to an average over $15$ instances of the same characteristics along with a $95\%$-confidence interval estimated following the $t$-distribution model. 

As expected, the larger the configuration cost $\theta$, the smaller the number of reconfiguration. The results show that for a desired number of reconfiguration, it is possible to chose an appropriate value of $\theta$ to enforce it. 
Also, it can be seen that the reconfiguration cost does not deteriorate dramatically the system's sharing policy performance. This means it is possible to reconfigure small subsets of paths on-the-fly (to the limit of what is feasible in terms of reconfiguration budget), and still enforce a satisfactorily fair policy, that will be ultimately optimally fair if the traffic requirement stabilizes. In the figures, we only showed the points for small values of $\theta$ -- this, along with the small precision level $\epsilon$, explain the apparition of a plateau below which it seems impossible to go. In reality, even larger values of $\theta$ permit one to accomplish a very low (down to zero) number of reconfigurations. 
We do not plot the points for larger values of $\theta$ in order to focus on the decrease of the reconfiguration number in the very beginning.

\section{Conclusion}

We designed an algorithm, FD-ADMM, that solves optimally the multi-path $\alpha$-fair resource allocation problem. We showed that FD-ADMM is fully distributed and that its implementation is suitable to distributed SDNs, regardless of the actual distribution of the networks into domains. The massively separable sub-problems given by FD-ADMM are efficiently solvable by the SDN domain controllers equipped with massively parallel hardware capable of addressing each update rule simultaneously as multiple-threads, thus reducing considerably the computation time and improving dramatically the responsiveness of the algorithm while providing in real-time feasible solutions. 
It was demonstrated numerically that the distribution of the SDN control permits to accelerate the overall system efficiency by attaining an equivalent optimality gap in highly reduced time. This shows that FD-ADMM scales naturally with the problems size, exploiting the computational power of modern SDN controllers. We also showed that the FD-ADMM extends easily to account for a switching cost per path, and that a trade-off between fairness and switching cost can be operated to preserve the system stability without deteriorating too much the resource allocation efficiency. In the future, we wish to specify the trade-off by providing theoretical bounds on the cost $\theta$ for this specific problem in order to respect any maximum allowed reconfiguration budget $\kappa$ in the multi-path setting.
\label{sec:conclusion}

\bibliographystyle{apa}

\end{document}